\documentclass[leqno,a4paper]{article}
\usepackage{amssymb,amsmath,amsthm,verbatim}

\title{\textbf{Structure of Gauge-Invariant Lagrangians}}

\author{\textsc{Marco Castrill\'on L\'opez}
\and \textsc{Jaime Mu\~{n}oz Masqu\'e}
\and \textsc{Eugenia Rosado Mar\'{\i}a}}

\date{}

\newtheorem{theorem}{Theorem}[section]

\theoremstyle{remark}
\newtheorem{remark}[theorem]{Remark}

\newtheorem{example}[theorem]{Example}

\begin{document}

\maketitle

\begin{abstract}
\noindent The theory of gauge fields in Theoretical Physics poses several
mathematical problems of interest in Differential Geometry and in Field Theory.
Below we tackle one of these problems: The existence of a finite system
of generators of gauge-invariant Lagrangians and how to compute them.
More precisely, if $p\colon C\to M$ is the bundle of connections on a principal
$G$-bundle $\pi\colon P\to M$, then a finite number $L_1,\dotsc,L_{N^\prime }$
of gauge-invariant Lagrangians defined on $J^1C$
is proved to exist such that for any other gauge-invariant Lagrangian
$L\in C^\infty (J^1C)$ there exists a function
$F\in C^\infty (\mathbb{R}^{N^\prime })$ such that $L=F(L_1,\dotsc,L_{N^\prime})$.
Several examples are dealt with explicitly.
\end{abstract}

\medskip

\noindent\emph{PACS numbers:\/} 2.20.Sv Lie algebras of Lie groups;
02.30.Ik Integrable systems;
02.40.-k Geometry, differential geometry, and topology;
03.50. z Classical field theories;
11.10.Ef Lagrangian and Hamiltonian approach;
11.15.-q Gauge field theories

\medskip

\noindent\emph{Mathematics Subject Classification 2000:\/} Primary 35F20;
Secondary 53C05,
58A20,
58D19,
58E15,
58E30,
81T13.

\medskip

\noindent\emph{Key words and phrases:} Bundle of connections, gauge
invariance, jet bundles, curvature mapping, functionally independent
gauge-invariant Lagrangians, structure of Lie algebras

\medskip

\noindent\emph{Acknowledgments:\/}
\textexclamdown \textexclamdown \textexclamdown \ Incluir proyecto quien lo
tenga !!!

\section{Introduction}
The notion of gauge invariance---defined as invariance under the group of
base-preserving automorphisms of a principal fibre $G$-bundle $\pi\colon
P\rightarrow M$ over an oriented space-time---is fundamental in the theory of
gauge fields and their associated fields, such as Yang-Mills-Higgs fields; for
example, see the classical expositions \cite{MM} or \cite{Traut}.

Below we are concerned with gauge invariance of the variational problems
determined by a free (i.e., without any interaction term with a particle
field) gauge invariant Lagrangian function defined on the fibre bundle
$p\colon C\rightarrow M$ of connections on $\pi\colon P\rightarrow M$.

The fundamental step in determining such Lagrangians, is the so-called
geometric formulation of Utiyama's theorem (see \cite[10.2.15 Theorem]{Bl}),
according to which a Lagrangian $L$ of first order on $C$ is gauge invariant
if and only if $L$ factors through the curvature map by means of a zero-order
Lagrangian on the vector bundle of differential $2$-forms on $M$ with values
in the adjoint bundle of $P$ (also called \textquotedblleft the curvature
bundle\textquotedblright), which, in turn, must be invariant under the natural
representation of the gauge group on that bundle. See \cite{Be} for the
generalization of Utiyama's theorem to Lagrangians for gauge-particle field interaction.

This reduces the problem of determining gauge-invariant Lagrangians to the
problem of determining the zero-order gauge invariant Lagrangians defined on
the curvature bundle. If $G$ is connected, then the second problem can
infinitesimally be solved by proving that zero-order gauge invariant
Lagrangians on the curvature bundle are the first integrals of an involutive
distribution $\mathcal{D}$.

In this paper, we prove that $\mathcal{D}$\ is of constant rank on a dense
open subset and we compute this rank. If $\mathfrak{g}$ is the Lie algebra of
$G$, and $n=\dim M$, $m=\dim\mathfrak{g}$, $l=\operatorname*{rank}%
\mathfrak{g}$, then we obtain the following results (see \S 3.2 below):

1st) If $n=2$, then the generic rank of $\mathcal{D}$ equals $m-l$, and

2nd) If $n\geq3$, then the generic rank is $m$.

The result for $n=2$ explains why the theory of Yang-Mills fields on a surface
presents special features. In fact, according to a classical theorem by
Chevalley (e.g., see \cite[Theorem 4.9.3]{Var}) there exist $l$ homogeneous
algebraically independent polynomials $p_{1},\ldots,p_{l}$ such that the
algebra of polynomial functions on $\mathfrak{g}$ that are invariant with
respect to the adjoint representation of $G$ on $\mathfrak{g}$, is isomorphic
to the algebra of polynomials in $p_{1},\ldots,p_{l}$, thus providing a basis
with geometric meaning for the algebra of first integrals of $\mathcal{D}$, as
stated in Remark 4.1.

Next, for $n\geq3$ we also obtain a basis of first integrals spanning
differentiably the ring of zero-order Lagrangians. Assuming $G$ connected and
semisimple, in the local case such a basis is deduced from Hilbert-Nagata
theorem (see Theorem 3.2).

Finally, we include several worked examples in low dimensions illustrating the
previous general results.
\section{Preliminaries}
\subsection{Jet prolongation}
Let $p\colon E\to M$ be a fibred manifold over an orientable connected
smooth manifold oriented by a volume form $\mathbf{v}$. A pair of
diffeomorphisms $\phi \in \mathrm{Diff}M$, $\Phi \in \mathrm{Diff}E$,
such that $p\circ \Phi =\phi \circ p$, is said to be an automorphism
of $p$; the group of all automorphisms is denoted by $\mathrm{Aut}E$.

If $Y\in \mathfrak{X}(E)$ then there exists a unique vector field
$Y^{(1)}$ in $\mathfrak{X}(J^1E)$---the $1$-jet prolongation
of $Y$ to the first-order jet bundle $J^1E$---projectable onto
$Y$ such that $\mathcal{L}_{Y^{(1)}}$ keeps invariant
the module of contact $1$-forms spanned by the following forms
$\theta ^\alpha =dy^\alpha -y_i^\alpha dx^i$, $1\leq \alpha\leq m$
on $\Omega^1(J^1E)$, where $n=\dim M$, $m+n=\dim E$, and $(x^i,y^\alpha )$,
$1\leq \alpha \leq m$, $1\leq i\leq n$, is a fibred
coordinate system for $p\colon E\to M$ and $(x^i,y^\alpha ,y_i^\alpha )$
is the induced coordinate system on $J^1E$.

The Lie algebra
of $\mathrm{Aut}E$ is the Lie subalgebra of $p$-projectable vector
fields $\mathfrak{X}_p(E)\subset \mathfrak{X}(E)$, namely, if $\Phi _t$
is the local flow of $Y\in \mathfrak{X}(E)$, then $\Phi_t \in \mathrm{Aut}E$,
$\forall t$, if and only if $Y\in \mathfrak{X}_p(E)$; in this case,
$\Phi _t^{(1)}$ is the flow of $Y^{(1)}$. If $Y$ is a $p$-vertical vector field,
then the formulas of $1$-jet prolongation are as follows:
\begin{equation}
\begin{array}
[c]{ll}
Y=v^\alpha \frac{\partial }{\partial y^\alpha }, & v^\alpha \in C^\infty (E),\\
Y^{(1)}=v^\alpha \frac{\partial }{\partial y^\alpha }
+v_i^\alpha \frac{\partial }{\partial y_i^\alpha }, & v_i^\alpha
=\frac{\partial v^\alpha }{\partial x^i}
+\frac{\partial v^\alpha }{\partial y^\beta }y_i^\beta ,
\end{array}
\label{i.c.t.}
\end{equation}
\subsection{$\mathrm{aut}P$ and $\mathrm{gau}P$\label{aut and Aut}}
Let $G$ be a Lie group. An automorphism of a principal $G$-bundle
$\pi\colon P\to M$ is a $G$-equivariant diffeomorphism $\Phi \colon P\to P$.
The group of all automorphisms of $P$ is denoted by $\mathrm{Aut}P$. Every
$\Phi \in \mathrm{Aut}P$ determines a unique diffeomorphism
$\phi \colon M\to M$, such that $\pi \circ \Phi =\phi \circ \pi $.
If $\phi $ is the identity map on $M$, then $\Phi $ is said to be
a gauge transformation (cf.\ \cite[3.2.1]{Bl}); the subgroup
of all gauge transformations is denoted by $\mathrm{Gau}P\subset \mathrm{Aut}P$.

A vector field $X\in \mathfrak{X}\left( P\right) $ is said to be
$G$-invariant if $R_{g}\cdot X=X$, $\forall g\in G$; if $\Phi_{t}$ is the flow
of $X$, then $X$ is $G$-invariant if and only if $\Phi_{t}\in\mathrm{Aut}P$,
$\forall t\in\mathbb{R}$. The Lie subalgebra of $G$-invariant vector fields on
$P$ is denoted by $\mathrm{aut}P\subset\mathfrak{X}\left( P\right) $. Each
$G$-invariant vector field on $P$ is $\pi $-projectable.

Similarly, a $\pi $-vertical vector field $X\in\mathfrak{X}\left( P\right) $
is $G$-invariant if and only if $\Phi _t\in \mathrm{Gau}P$, $\forall
t\in \mathbb{R}$. Let $\mathrm{gau}P\subset\mathrm{aut}P$ be the ideal of all
$\pi $-vertical $G$-invariant vector fields on $P$, which is usually called the
gauge algebra of $P$.

The quotient $T(P)/G$ exists as a differentiable manifold and it is endowed
with a vector bundle structure over $M$ (see \cite{Atiyah}), whose global
sections can naturally be identified to $\mathrm{aut}P$; i.e.,
$\mathrm{aut}P\cong \Gamma \left( M,T(P)/G\right) $.

If $G$ acts on the left on a manifold $F$ via a map $G\times F\to F$,
$(g,y)\mapsto g\star y$, then $G$ acts on the right
on the product $P\times F$\ by setting
$(u,y)\cdot g=(u\cdot g,g^{-1}\star y)$, $\forall g\in G$,
$\forall u\in P$, $\forall y\in F$.
The quotient manifold $(P\times F)/G$ of this action exists and it defines
a fibre bundle $\pi _F\colon P\times ^GF\to M$,
$\pi _F((u,y)\operatorname{mod}G)=\pi (u)$, called the bundle associated
to $P$ by the action on $F$; e.g., see \cite[\S 3.1]{Bl}, \cite[\S 35]{GS},
\cite[p.54]{KN}.

Every $\phi \in \mathrm{Aut}P$ induces a diffeomorphism
$\Phi _F\colon P\times ^GF\to P\times ^GF$ by setting
$\Phi_{F}((u,y)\operatorname{mod}G)=(\Phi (u),y)\operatorname{mod}G$,
$\forall u\in P$, $\forall y\in F$. If $\phi \in \mathrm{Diff}M$
is the diffeomorphism determined by $\Phi $, then
$\pi _F\circ \Phi _F=\phi \circ \pi _F$.

Smooth sections $s\colon M\to P\times ^GF$ of $\pi _F\colon P\times ^GF\to M$
are in one-to-one correspondence with $G$-equivariant smooth maps
$\hat{s}\colon P\to F$, which means in this case that
$\hat{s}(u\cdot g)=g^{-1}\star \hat{s}(u)$, $\forall u\in P$, $\forall g\in G$,
see \cite[Proposition 35.1]{GS}.

If $F$\ is a Lie group and every left translation $L_g\colon F\to F$
of the action of $G$\ on $F$\ is an automorphism, then the fibres of
$\pi _F\colon P\times ^GF\to M$\ are endowed with a structure of Lie
group isomorphic to $F$.\ In fact, we can represent two given points
$f,f^\prime \in \pi _F^{-1}(x)$ as $f=(u,y)\operatorname{mod}G$,
$f^\prime =(u,y^\prime )\operatorname{mod}G$, with the same $u\in P$,
as $G$ acts transitively on $\pi ^{-1}(x)$, and thus the operation given
by $f\cdot f^\prime =(u,y\cdot y^\prime )\operatorname{mod}G$, makes sense.
Hence $\pi _F$\ is Lie-group fibre bundle.

In particular, if $G$ acts on itself by conjugation, $G\times G\to G$,
$(g,x)\to gxg^{-1}$, $\forall g,x\in G$, then the associated bundle
$\pi _G\colon \mathrm{Ad}P=P\times ^GG\to M$ is a Lie-group fibre
bundle. Consequently, smooth sections
$\Gamma (M,\mathrm{Ad}P)$ of $\pi _G$ admit a group structure given by
$(s\cdot s^\prime )(x)=s(x)\cdot s^\prime (x)$, for all
$s,s^\prime \in \Gamma (M,\mathrm{Ad}P)$, $x\in M$, and the group
$\mathrm{Gau}P$ is isomorphic to $\Gamma(M,\mathrm{Ad}P)$, see
\cite[Proposition 35.2]{GS}.

Similarly $\mathrm{gau}P\cong \Gamma \left( M,\mathrm{ad}P\right) $,
where $\pi _{\mathfrak{g}}\colon \mathrm{ad}P\to M$ denotes the adjoint
bundle, i.e., the bundle associated to $P$ by the adjoint representation
of $G$ on its Lie algebra $\mathfrak{g}$; namely $\mathrm{ad}P
=\left( P\times \mathfrak{g}\right) /G$, the action of $G$
on $P\times \mathfrak{g}$ being defined by $\left( u,A\right) \cdot g
=(u\cdot g,\mathrm{Ad}_{g^{-1}}\left( A\right) )$,
$\forall u\in P$, $\forall A\in \mathfrak{g}$, $\forall g\in G$.
The $G$-orbit of $\left( u,A\right) \in P\times\mathfrak{g}$ in
$\mathrm{ad}P$ is denoted by $(u,A)_{\mathrm{ad}}$. We thus obtain an exact
sequence of vector bundles over $M$ (the so-called Atiyah sequence, see
\cite[Th.\ 1]{Atiyah}):
\begin{equation}
\label{Atiyah_seq}
0\to\mathrm{ad}P\to T(P)/G
\overset {\pi _{\ast}}{\longrightarrow }TM\to 0.
\end{equation}
The fibres $\left( \mathrm{ad}P\right) _x$ are endowed
with a Lie algebra structure determined by
\begin{equation}
\left[ (u,A)_{\mathrm{ad}},(u,B)_{\mathrm{ad}}\right]
=(u,-\left[ A,B\right] )_{\mathrm{ad}},\quad
\forall u\in \pi ^{-1}\left( x\right ), \,
\forall A,B\in \mathfrak{g,}
\label{f2}
\end{equation}
where $\left[ \cdot ,\cdot \right] $ denotes the bracket in $\mathfrak{g}$,
but this is no longer true for the fibres of $T(P)/G$. The sign of the bracket
in \eqref{f2} is needed in order to ensure that the natural identification
$\mathrm{gau}P\simeq \Gamma \left( M,\mathrm{ad}P\right) $ is a Lie algebra
isomorphism, when $\mathrm{gau}P$ is considered as a Lie subalgebra of
$\mathfrak{X}(P)$.
\section{Bundle of connections\label{bundle_connections}}
Let $X^{h_\Gamma }\in \mathfrak{X}\left( P\right) $ be the horizontal lift
of $X\in \mathfrak{X}\left( M\right) $ with respect to a connection $\Gamma $
on $\pi \colon P\to M$. The vector field $X^{h_\Gamma }$ is
$G$-invariant and projects onto $X$ (cf.\ \cite[II.\ Proposition\ 1.2]{KN}).
Hence we have a splitting of \eqref{Atiyah_seq}, $s_\Gamma \colon TM
\to T(P)/G$, $s_\Gamma \left(  X\right)  =X^{h_\Gamma }$.
Conversely, any splitting $\sigma\colon TM\to T(P)/G$ of that sequence
comes from a unique connection on $P$. Therefore there is a natural bijection
between connections on $P$ and splittings of the sequence above. Connections
on $P$ can be identified to the global sections of a bundle $p\colon C\to M$;
the section of $p$ induced by $\Gamma$ is denoted by
$s_ \Gamma \colon M\to C$. Moreover, $C$ is an affine bundle modelled over
$\mathrm{Hom}\left( TM,\mathrm{ad}P\right) \cong T^\ast M\otimes \mathrm{ad}P$.
For more details we refer the reader to \cite{CM}.

Let $\left( U;x^i\right) $ be a coordinate system on an open domain
$U\subset M$ over which $\pi $ admits a section $s\colon U\to P$, so
that $\pi^{-1}(U)\cong U\times G$. For every $B\in \mathfrak{g}$ let
$\tilde{B}$ be the infinitesimal generator of the flow of gauge transformations
over $U$ defined by $\varphi_{t}^{B}\left( x,g\right)
=\left( x,\exp\left( tB\right) \cdot g\right) $, $x\in U$.
As $\pi\circ\varphi_{t}^{B}=\pi$ the vector field
$\tilde{B}\in \mathfrak{X}(\pi ^{-1}(U))$ is $\pi $-vertical. If
$(B_1,\dotsc,B_m)$ is a basis of $\mathfrak{g}$, then
$\tilde{B}_1,\dotsc,\tilde{B}_m$ is a basis of $\Gamma(U,\mathrm{ad}P)$.
The horizontal lift with respect to $\Gamma $ of the basic vector field
$\partial/\partial x^i$ is given as follows:
\begin{equation}
s_\Gamma \Bigl( \frac{\partial }{\partial x^i}\Bigr)
=\Bigl( \frac{\partial }{\partial x^i}\Bigr) ^{h_\Gamma }
=\frac{\partial }{\partial x^i}
-\left( A_i^\alpha \circ s_\Gamma \right) \tilde{B}_\alpha ,
\quad 1\leq i\leq n.
\label{f10}
\end{equation}
The functions $(x^i,A_{j}^\alpha )$, $i,j=1,\dotsc,n=\dim M$, $1\leq
\alpha\leq m=\dim G$, induce a coordinate system on $p^{-1}\left(  U\right)
=C\left(  \pi^{-1}U\right)  $ (cf.\ \cite{CM}); hence $\dim C=n(m+1)$.

Each automorphism $\phi \in \mathrm{Aut}P$ acts on connections by pulling back
connection forms; i.e., $\Gamma ^\prime =\Phi \left( \Gamma\right) $ where
$\omega _{\Gamma ^\prime }=(\Phi ^{-1})^\ast \omega _\Gamma $ (cf.\ \cite[II.\
Proposition\ 6.2-(b)]{KN}). For each $\phi \in \mathrm{Aut}P$ there exists a
unique diffeomorphism $\Phi _C\colon C\to C$ such that
$p\circ \Phi _C=\phi \circ p$, where $\phi \colon M\to M$ is the
diffeomorphism induced by $\Phi$ on the ground manifold. We thus obtain a
group homomorphism $\mathrm{Aut}P\to \mathrm{Diff}C$ and for every
connection $\Gamma$ on $P$ we have $\Phi _C\circ s_\Gamma
=s_{\Phi \left( \Gamma\right) }$. If $\Phi _t$ is the flow of a $G$-invariant
vector field $X\in \mathrm{aut}P$, then $\left( \Phi _t\right) _C$
is a one-parameter group in $\mathrm{Diff}C$ with infinitesimal generator
denoted by $X_C$, and the map $\mathrm{aut}P\to \mathfrak{X}\left( C\right) $,
$X\mapsto X_C$ is a Lie-algebra homomorphism.

By using a coordinate domain $\left(  U;x^i\right)  _{i=1}^{n}$ in $M$ and
the basis $(\tilde{B}_\alpha )_{\alpha=1}^{m}$ of $\mathrm{ad}\pi^{-1}(U) $
introduced above, it follows that each $X\in \mathrm{gau}\pi^{-1}(U)$ can be
written as
\begin{equation}
X=g^\alpha \tilde{B}_\alpha ,\quad g^\alpha \in C^\infty (U),
\label{f20}
\end{equation}
and, as a computation shows (e.g., see \cite{CM}), we have
\begin{equation}
X_C=-\left( \frac{\partial g^\alpha }{\partial x^i}
-c_{\beta \gamma }^\alpha g^\beta A_i^\gamma \right)
\frac{\partial }{\partial A_i^\alpha },
\label{f21}
\end{equation}
where $c_{\beta\gamma}^\alpha $ are the structure constants:
$[B_\beta ,B_\gamma ]=c_{\beta \gamma }^\alpha B_\alpha $.
\section{Gauge invariance}
\subsection{$\mathrm{gau}P$ and $\mathrm{aut}P$ invariance}
A Lagrangian density $\Lambda=Lv$, $L\in C^\infty (J^1C)$,
on the bundle of connections is said to be gauge invariant
if $X_C^{(1)}(L)=0$, $\forall X\in \mathrm{gau}P$, where
$X_C^{(1)}$ denotes the $1$-jet prolongation of the natural
representation $\mathrm{aut}P\to \mathfrak{X}(C)$. Similarly,
a Lagrangian density is said to be $\mathrm{aut}P$-invariant if
\[
\mathcal{L}_{X_C^{(1)}}(Lv)=X_C^{(1)}(L)v+L(\mathcal{L}_{X_C^{(1)}}v)=0,
\quad \forall X\in \mathrm{aut}P.
\]
The vector field $X_C^{(1)}$ is $p_1$-projectable onto
$X^\prime =\pi _\ast X$, where $p_1\colon J^1C\to M$ is the
canonical projection. We thus have
$\mathcal{L}_{X_C^{(1)}}(Lv)=(X_C^{(1)}(L)+L\mathrm{div}X^\prime )v$
and the condition of $\mathrm{aut}P$-invariance yields
$X_C^{(1)}(L)+L\mathrm{div}X^\prime =0$, $\forall X\in \mathrm{aut}P$.
It turns out, every $\mathrm{aut}P$-invariant Lagrangian
density is variationally trivial, as it is proved in
\cite[Corollary 1]{CMT}. Thus, the notion of $\mathrm{aut}P$-invariance
is too restrictive to be useful in Field Theory.

If $X\in \mathrm{gau}P$, then $X^\prime =0$ and the definition
of gauge invariance is recovered. As every $\Phi \in \mathrm{gau}P$
induces the identity map on $M$, the function $L$ is gauge invariant
if and only if the gauge group is a group of symmetries
of the Lagrangian density $\Lambda =L\mathbf{v}$, where $\mathbf{v}$
is the volume form on the ground manifold. For more details we refer
the reader to \cite{CM}, \cite{CMT}, and \cite{MJ}.
\subsection{The number of gauge-invariant Lagrangians}
Let
\begin{equation}
\begin{array}
[c]{l}
\Omega \colon J^1C\to \bigwedge\nolimits^2T^\ast M\otimes \mathrm{ad}P\\
\Omega (j_x^1\sigma _\Gamma )=\left( \Omega _\Gamma \right) _x
\end{array}
\label{curvature}
\end{equation}
be the curvature mapping. The curvature form $\Omega _\Gamma $ of the
connection $\Gamma $ corresponding to a section $s_\Gamma $ of $p$ is seen
to be a two form on $M$ with values in the adjoint bundle $\mathrm{ad}P$.
On the vector bundle $\bigwedge\nolimits^2T^\ast M\otimes \mathrm{ad}P$
we consider the coordinate systems $(x^i;R_{jk}^\alpha )$, $j<k$,
induced by a coordinate system $(U;x^i)_{i=1}^n$ on $M$, and a basis
$(B_\alpha )_{\alpha =1}^m$ of $\mathfrak{g}$, as follows:
\begin{equation}
\eta _2=\sum \limits_{j<k}R_{jk}^\alpha (\eta _2)
\left( dx^j\wedge dx^k\otimes \tilde{B}_\alpha \right) _x,\;
\forall\eta _2\in \bigwedge\nolimits^2T_x^\ast M\otimes (\mathrm{ad}P)_x,
\;\forall x\in U.\label{R's}
\end{equation}
With respect to the coordinate systems $(x^i,A_j^\alpha ,A_{j,k}^\alpha )$
and $(x^i;R_{jk}^\alpha )$, $j<k$, on $J^1C$ and
$\bigwedge\nolimits^2T^\ast M\otimes \mathrm{ad}P$, respectively,
the equations of the curvature mapping are as follows:
\begin{equation}
\begin{array}
[c]{rl}
R_{jk}^\alpha \circ\Omega = & A_{j,k}^\alpha -A_{k,j}^\alpha
-\sum _{\beta <\gamma }c_{\beta \gamma}^\alpha
\left( A_j^\beta A_k^\gamma -A_j^\gamma A_k^\beta \right) ,
\smallskip\\
& 1\leq j<k\leq n,\quad 1\leq \alpha \leq m.
\end{array}
\label{eqs1}
\end{equation}

The geometric formulation of Utiyama's Theorem (e.g., see \cite{Bl})
states that a Lagrangian $L\colon J^1C\to \mathbb{R}$ is gauge invariant
if and only if $L$ factors through $\Omega $ as $L=\tilde{L}\circ \Omega$,
where
\begin{equation}
\tilde{L}\colon \bigwedge\nolimits^2T^\ast M\otimes \mathrm{ad}P
\to \mathbb{R}
\label{Lbarra}
\end{equation}
is a $C^\infty $ function that is invariant under the adjoint representation
of $G$ on the curvature bundle. As the curvature map \eqref{curvature} is
surjective, the function $\tilde{L}$\ is unique.

If the group $G$ is connected, then according to the formulas \eqref{f21} and
\eqref{i.c.t.}, and taking the equations of the curvature mapping \eqref{eqs1}
into account, the function $\tilde{L}$ in the formula \eqref{Lbarra} is
invariant under the adjoint representation of $G$ on the curvature bundle if
and only if
\begin{equation}
\begin{array}
[c]{rll}
\chi _\alpha (\tilde{L})= & 0, & 1\leq \alpha \leq m.\\
\chi _\alpha = &
{\displaystyle\sum\nolimits_{i<j}}
c_{\gamma \alpha }^\beta R_{ij}^ \gamma
\dfrac{\partial }{\partial R_{ij}^\beta }. &
\end{array}
\label{inv_ad}
\end{equation}
Alternatively, this equivalence can also be deduced
from the formula for $X_C^{(1)}$ in \cite[(2.10)]{CMT}.
\begin{theorem}
\label{th1}
Assume the group $G$ is connected.

The distribution $\mathcal{D}$ on
$\bigwedge\nolimits^2T^\ast M\otimes \mathrm{ad}P$ generated
by the vector fields $\chi _\alpha $, $1\leq\alpha\leq m$,
given in \emph{\eqref{inv_ad}}, is involutive.
\begin{enumerate}
\item[\emph{(i)}]
The Lie algebra $\mathfrak{g}$ is Abelian if and only if
$\mathcal{D}=\{ 0\} $.

\item[\emph{(ii)}]
If $\mathfrak{g}$ is not Abelian, then the rank of
$\mathcal{D}$ on a dense open subset is constant.

\item[\emph{(iii)}]
If $\dim M=n=2$ and $\dim\mathfrak{g}=m$,
$\operatorname*{rank}\mathfrak{g}=l$, then the generic rank of $\mathcal{D}$
is equal to $m-l$.

\item[\emph{(iv)}]
If $\dim M=n\geq3$, $\dim\mathfrak{g}=m$, and
$\mathfrak{g}$ is semisimple, then the generic rank of $\mathcal{D}$ is equal
to $m$.
\end{enumerate}
\end{theorem}
\begin{proof}
As a computation shows, we have
$[\chi _\rho ,\chi _\sigma ]=c_{\rho \sigma }^\gamma \chi _\gamma $
for $1\leq \rho <\sigma \leq m$; hence $\mathcal{D}$ is involutive.

(i) The vector fields $\chi _\alpha $, $1\leq \alpha \leq m$, vanish if and only
if $c_{\gamma \alpha }^\beta R_{ij}^\gamma=0$ for all $\alpha ,\beta =1,\dotsc,m$,
$1\leq i<j\leq n$, and these equations are obviously equivalent
to saying that $c_{\gamma \alpha }^\beta =0\,$for all
$\alpha ,\beta ,\gamma=1,\dotsc,m$.
\smallskip

(ii) The rank of $\mathcal{D}$ at a point
$\eta _2\in \bigwedge\nolimits^2T_{x_0}^\ast M\otimes (\mathrm{ad}P)_{x_0}$
equals the rank of the $m\times m\tbinom{n}{2}$ matrix
\[
\begin{array}
[c]{rll}
\Lambda (\eta _2)= & \left( \lambda_{\beta ,ij}^\alpha (\eta _2)\right)
_{1\leq \beta \leq m,1\leq i<j\leq n}^{1\leq \alpha \leq m},
& \lambda_{\beta ,ij}^\alpha (\eta _2)
=c_{\gamma\alpha}^{\beta}R_{ij}^{\gamma}(\eta _2).
\end{array}
\]
As the entries of $\Lambda $ are polynomial functions in the coordinates
$R_{ij}^\alpha $, $1\leq\alpha\leq m$, $1\leq i<j\leq n$, it follows that the
rank of $\mathcal{D}$ takes its maximum value on a dense open subset, and we
have $\max _{\eta _2}\operatorname*{rank}\Lambda (\eta _2)\leq m$,
$\forall\eta _2\in \bigwedge\nolimits^2T^\ast M\otimes \mathrm{ad}P$.
\smallskip

(iii) If $n=2$, then $\Lambda (\eta _2)$ is a square matrix of size $m$.
If $B=R_{12}^{\gamma}(\eta _2)B_\gamma $, then the matrix of the linear map
$\mathrm{ad}_B\colon \mathfrak{g}\to \mathfrak{g}$ in the basis
$(B_\gamma )_{\gamma =1}^m$ coincides with $\Lambda (\eta _2)$; in fact,
we have $\mathrm{ad}_B(B_\alpha )=[B,B_\alpha ]
=c_{\gamma \alpha}^\beta R_{12}^\gamma (\eta _2)B_\beta $,
$1\leq\alpha\leq m $. Hence, if $B\in \mathfrak{g}$ is a regular element,
then the rank of $\Lambda (\eta _2)$ is $m-l$ exactly.
\smallskip

(iv) For every pair of indices $1\leq j<k\leq n$, let $\Lambda_{jk}(\eta _2)$
be the $m\times m$ matrix $\Lambda_{jk}(\eta _2)=\left(  \lambda_{\beta
,jk}^\alpha (\eta _2)\right)  _{\alpha,\beta=1}^{m}$. Then the $m\times
m\tbinom{n}{2}$ matrix $\Lambda(\eta _2)$ can be written in blocks as
follows: $\Lambda(\eta _2)=\left(  \Lambda_{12}(\eta _2),\dotsc,\Lambda
_{1n}(\eta _2),\dotsc,\Lambda_{n-1,n}(\eta _2)\right)  $.

In order to prove this case, we can use a Chevalley basis (e.g., see
\cite[Chapter 3, Theorem 1.19]{GOV}); more precisely: Let
$\{ \alpha _1,\dotsc,\alpha _l\} $ be a system of simple roots in the set
$\Delta _{\mathfrak{g}}=\{ \alpha _1,\dotsc,\alpha _l,\alpha _{l+1},\dotsc,
\alpha _{m-l}\}$, and let $h_i=h_{\alpha _i}$ for $1\leq i\leq l$. The
basis $h_i$, $1\leq i\leq l$, $e_\alpha $, $\alpha \in \Delta _{\mathfrak{g}}$,
of $\mathfrak{g}$ satisfies the following properties: $[h_i,h_j]=0$ for
$i,j=1,\dotsc,l$, $[h_i,e_\alpha ]=\left\langle \alpha|\alpha _i\right\rangle
e_\alpha $, $\alpha \in \Delta _{\mathfrak{g}}$, $1\leq i\leq l$, and
$[e_\alpha ,e_{-\alpha}]=h_\alpha =e^ih_i$, $e^i\in \mathbb{Z}$,
$1\leq i\leq l$, and if $\alpha ,\beta \in \Delta _{\mathfrak{g}}$,
$\alpha +\beta \neq0$, and $\beta -p\alpha ,\dotsc,\beta +q\alpha $
is the $\alpha $-string of roots containing $\beta $, then
$[e_\alpha ,e_\beta ]=0$ if $q=0$, $[e_\alpha ,e_\beta ]
=\pm(p+1)e_{\alpha +\beta}$ if $\alpha +\beta \in \Delta _{\mathfrak{g}}$;
or equivalently
\[
\begin{array}
[c]{rlrlll}
c_{ij}^{k}= & \!\!0, & c_{ij}^\alpha = & \!\!0, & i,j=1,\dotsc,l, &
\!\!\alpha \in \Delta _{\mathfrak{g}},\\
c_{i\alpha }^k= & \!\!0, & c_{i\alpha }^\beta = & \!\!
\delta_\alpha ^\beta
\left\langle \alpha|\alpha_i\right\rangle , & 1\leq i\leq l, &
\!\!\alpha,\beta \in \Delta _{\mathfrak{g}},\\
c_{\alpha,-\alpha}^i= & \!\! e^i, & c^\beta _{\alpha ,-\alpha }= &
\!\!0, & 1\leq i\leq l, & \!\!\alpha, \beta \in \Delta _{\mathfrak{g}},\\
c_{\alpha\beta}^i= & \!\!0, & c_{\alpha \beta}^\gamma = & \!\!0 &
\alpha ,\beta ,\gamma \in \Delta _{\mathfrak{g}}, & \!\!\alpha +\beta \neq 0,
\alpha +\beta \notin \Delta _{\mathfrak{g}},\\
c_{\alpha \beta }^i= & \!\!0, & c_{\alpha \beta }^\gamma = & \!\!
\pm \delta _{\alpha +\beta }^\gamma (p+1), & 1\leq i\leq l, & \!\!
\alpha + \beta \in \Delta _{\mathfrak{g}},\gamma \in \Delta _{\mathfrak{g}}.
\end{array}
\]
According to the general notations \eqref{R's} in this case we can write
\[
\eta _2=\sum\limits_{j<k}\sum _{i=1}^l\left( R_{jk}^i(\eta _2)dx^j
\wedge dx^k\otimes \tilde{h}_i\right) _x
+\sum\limits_{j<k}\sum _{a=1}^{m-l}\left( R_{jk}^{\alpha _a}(\eta _2)
dx^j\wedge dx^k\otimes \tilde{e}_{\alpha _a}\right) _x.
\]
With these notations, for $1\leq i<j\leq n$ we have
$\lambda _{u,jk}^t(\eta _2)=0$ for all $t,u=1,\dotsc,l$, and
\begin{equation}
\begin{array}
[c]{rll}
\lambda _{\beta _b,jk}^t(\eta _2)= &
-\left\langle \beta _b|\beta_t\right\rangle
R_{jk}^{\beta _b}(\eta _2), & 1\leq t\leq l,1\leq b\leq m-l,
\end{array}
\label{t_beta}
\end{equation}
\begin{equation}
\begin{array}
[c]{rll}
\lambda_{t,jk}^{\alpha_{a}}(\eta _2)= & -e^{t}R_{jk}^{-\alpha_{a}}(\eta
_2), & 1\leq t\leq l,1\leq a\leq m-l,
\end{array}
\label{alpha_t}
\end{equation}
\begin{equation}
\begin{array}
[c]{rl}
\lambda _{\alpha _v,jk}^{\alpha _u}(\eta _2)= &
\sum _{t=1}^l\delta _{\alpha _v}^{\alpha _u}
\left\langle \alpha _v|\alpha _t\right\rangle
R_{jk}^{t}(\eta _2)\\
& \multicolumn{1}{r}{\mp\sum _{r=1}^{m-l}(p_r+1)
\delta _{\alpha _r+\alpha _v}^{\alpha _u}
R_{jk}^{\alpha _r}(\eta _2),
\smallskip }\\
& \multicolumn{1}{r}{u,v=1,\dotsc,m-l.}
\end{array}
\label{alpha_u_alpha_v}
\end{equation}
Distinguishing cases, \eqref{alpha_u_alpha_v} is readily seen
to be equivalent to the following two formulas:

If $\alpha _u-\alpha _v=\alpha _{r_0}\in \Delta _{\mathfrak{g}}$,
for some $1\leq r_0\leq m-l$, then
\begin{equation}
\begin{array}
[c]{ll}
\lambda _{\alpha _v,jk}^{\alpha _u}(\eta _2)= &
\sum_ {t=1}^l\delta _{\alpha _v}^{\alpha _u}
\left\langle \alpha _v|\alpha _t\right\rangle
R_{jk}^{t}(\eta _2)\\
\multicolumn{1}{r}{} &
\multicolumn{1}{r}{\mp(p_{r_0}+1)R_{jk}^{\alpha _{r_0}}(\eta _2),
\smallskip}\\
\multicolumn{1}{r}{} & \multicolumn{1}{r}{u,v=1,\dotsc,m-l.}
\end{array}
\label{alpha_u_alpha_v_1}
\end{equation}
If $\alpha _u-\alpha _v\notin \Delta _{\mathfrak{g}}$, then
\begin{equation}
\begin{array}
[c]{ll}
\lambda _{\alpha _v,jk}^{\alpha _u}(\eta _2)= &
\sum _{t=1}^l\delta _{\alpha _v}^{\alpha _u}
\left\langle \alpha _{v}|\alpha _{t}\right\rangle
R_{jk}^t(\eta _2),
\smallskip\\
\multicolumn{1}{r}{} & \multicolumn{1}{r}{u,v=1,\dotsc,m-l.}
\end{array}
\label{alpha_u_alpha_v_2}
\end{equation}
Hence the $m\times m$ matrix $\Lambda _{jk}(\eta _2)$ is given by
\[
\Lambda _{jk}(\eta _2)=\left(
\begin{array}
[c]{cc}
O & A_{jk}(\eta _2)\\
B_{jk}(\eta _2) & C_{jk}(\eta _2)
\end{array}
\right)  ,
\]
where $O$ denotes the $l\times l$ zero matrix, and $A_{jk}(\eta _2)$,
$B_{jk}(\eta _2)$, $C_{jk}(\eta _2)$ are the matrices with sizes
$l\times(m-l)$, $(m-l)\times l$, $(m-l)\times(m-l)$, respectively, given by
\[
A_{jk}(\eta _2)=\left( \lambda _{\alpha _a,jk}^t(\eta _2)
\right) _{1\leq a\leq m-l}^{1\leq t\leq l},
\qquad
B_{jk}(\eta _2)=\left( \lambda _{t,jk}^{\alpha _a}(\eta _2)
\right) _{1\leq t\leq l}^{1\leq a\leq m-l},
\]
\[
C_{jk}(\eta _2)=\left( \lambda _{a_b,jk}^{\alpha _a}(\eta _2)
\right) _{1\leq b\leq m-l}^{1\leq a\leq m-l}.
\]
According to \eqref{t_beta} we have
\[
A_{jk}(\eta _2)=-\left(
\begin{array}
[c]{ccc}
\left\langle \beta_1|\beta_1\right\rangle R_{jk}^{\beta _1}(\eta _2) &
\ldots & \left\langle \beta _{m-l}|\beta _1\right\rangle
R_{jk}^{\beta _{m-l}}(\eta _2)\\
\vdots & \ddots & \vdots \\
\left\langle \beta _1|\beta _l\right\rangle
R_{jk}^{\beta _1}(\eta _2) & \ldots &
\left\langle \beta _{m-l}|\beta _l\right\rangle
R_{jk}^{\beta _{m-l}}(\eta _2)
\end{array}
\right) ,
\]
and according to \eqref{alpha_t} we have
\[
B_{jk}(\eta _2)=-\left(
\begin{array}
[c]{ccc}
e^1R_{jk}^{-\alpha _1}(\eta _2) & \ldots & e^lR_{jk}^{-\alpha _1}(\eta _2)\\
\vdots & \ddots & \vdots\\
e^1R_{jk}^{-\alpha _{m-l}}(\eta _2) &
\ldots & e^lR_{jk}^{-\alpha _{m-l}}(\eta _2)
\end{array}
\right) ,
\]
and $C_{jk}(\eta _2)=
(\lambda _{\alpha _v,jk}^{\alpha _u}(\eta _2))_{1\leq v\leq m-l}^{1\leq u\leq m-l}$
is given by the formulas \eqref{alpha_u_alpha_v_1} and \eqref{alpha_u_alpha_v_2}.

Let $A\subset \Delta _{\mathfrak{g}}$ be the set of elements
$\alpha _{r_0}\in \Delta _{\mathfrak{g}}$ for which there exist
$\alpha _u,\alpha _v\in \Delta _{\mathfrak{g}}$ such that
$\alpha _u-\alpha _v=\alpha _{r_0}$,
and let $E_{x_0}\subset \bigwedge\nolimits^2T_{x_0}^\ast M\otimes
(\mathrm{ad}P)_{x_0}$ be the closed subset of
$(\mathrm{ad}P)_{x_0}$-valued $2$-covectors $\eta _2^0$ such that
$R_{ij}^{\alpha _{r_0}}(\eta _2^0)=0$ for all indices
$1\leq i<j\leq n$, as long as $\alpha _{r_0}\in A$. If
$\eta _2^0\in E_{x_0}$, then $C_{ij}(\eta _2^0)$ is
a diagonal square matrix of order $m-l$, whose non-vanishing entries
are given by
\[
\begin{array}
[c]{rlll}
\mu _{u,ij}=\lambda _{\alpha _u,ij}^{\alpha _u}(\eta _2)= &
\sum _{t=1}^l\left\langle \alpha _u|\alpha _t\right\rangle
R_{ij}^t(\eta _2^0), & 1\leq u\leq m-l, & 1\leq i<j\leq n.
\end{array}
\]
Hence, by taking the values $R_{ij}^t(\eta _2^0)$
in a suitable dense open subset in $E_{x_0}$,
it follows that $\det C_{ij}(\eta _2^0)\neq 0$.

We have $m>2l$ (see Remark \ref{remark2}), as $\mathfrak{g}$ is semisimple.
We can thus decompose the matrix $A_{jk}(\eta _2)$
into two blocks as follows: $A_{jk}(\eta _2)
=(A_{jk}^\prime (\eta _2),A_{jk}^{\prime\prime}(\eta _2))$,
where
$A_{jk}^\prime (\eta _2)$ denotes the $l\times l$ submatrix
of $A_{jk}(\eta _2)$ determined by its $l$ rows and its first $l$ columns
of $A_{jk}(\eta _2)$, whereas $A_{jk}^{\prime\prime}(\eta _2)$ denotes the
$l\times(m-2l)$ submatrix of $A_{jk}(\eta _2)$ determined by its $l$ rows and
its $m-2l$ columns. As a computation shows, we have
\[
\det A_{jk}^\prime (\eta _2^0)=R_{jk}^{\alpha _1}(\eta _2^0)\cdots
R_{jk}^{\alpha _{l}}(\eta _2^{0})
\det (\left\langle \alpha _u|\alpha _v\right\rangle )_{u,v=1}^l.
\]
Hence $A_{jk}^\prime (\eta _2^0)$ is non-singular on a dense open subset
in $E_{x_0}$.

Moreover, since $n\geq 3$, we can consider the $m\times2m$ submatrix of
$\Lambda (\eta _2^0)$ given by $\Lambda ^\prime (\eta _2^0)=\left(
\Lambda _{12}(\eta _2^0),\Lambda _{13}(\eta _2^0)\right) $, and also the
$m\times m$ submatrix $\Lambda ^{\prime \prime }(\eta _2^0)$ of
$\Lambda ^\prime (\eta _2^0)$ defined by
\[
\Lambda ^{\prime \prime }(\eta _2^0)=\left(
\begin{array}
[c]{cc}
A_{12}(\eta _2^0) & A_{13}^\prime (\eta _2^0)\\
C_{12}(\eta _2^0) & C_{13}^\prime (\eta _2^0)
\end{array}
\right) ,
\]
where $C_{13}^\prime (\eta _2^0)$ is the $l\times l$ matrix
\[
C_{13}^\prime (\eta _2^0)=\left(
\begin{array}
[c]{cccc}
\mu _{1,13} & 0 & \ldots & 0\\
0 & \mu _{2,13} & 0 & 0\\
\vdots & 0 & \ddots & \vdots\\
0 & 0 & \ldots & \mu_{l,13}
\end{array}
\right) .
\]
Next, the determinant of $\Lambda ^{\prime \prime }(\eta _2^0)$ is evaluated
by the Laplacian expansion along the $l\times l$ minors of its $l$ last
columns; e.g., see \cite[III, \S 8, formula\ (21)]{Bourb}.

Let $S_l$ be the set of $l$-element subsets of $\{ 1,2,...,m\} $ and for every
$I=\{ i_1<\ldots <i_l\} \in S_l$ let us denote by $\Delta _I$ the $l\times l$
submatrix of
\[
\left(
\begin{array}
[c]{c}
A_{13}^\prime (\eta _2^0)\\
C_{13}^\prime (\eta _2^0)
\end{array}
\right)
\]
determined by the rows $i_1,\dotsc,i_l$; for example
\[
\begin{array}
[c]{rlrl}
\Delta _{\{1,\dotsc,l\}}= & A_{13}^\prime (\eta _2^0), &
\Delta _{\{m-l+1,\dotsc,m\}}= & C_{13}^\prime (\eta _2^0).
\end{array}
\]
If $\Delta _I^\prime $ denotes the complement of $\Delta _I$ in
$\Lambda ^{\prime \prime }(\eta _2^0)$, then by setting
$|I|=i_1+\ldots +i_l$, we have
$\det \Lambda ^{\prime \prime }(\eta _2^0)
=\sum _{I\in S_l}(-1)^{\frac{(m-l+1)+m}{2}l+|I|}(\det \Delta _I)
(\det \Delta _I^\prime )$.

In this formula the functions $\det \Delta _I$ depend on the values
$R_{13}^t(\eta _2^0)$ only, while the functions $\det \Delta _I^\prime $
depend on the values $R_{12}^t(\eta _2^0)$ only. Hence
$\det \Lambda ^{\prime \prime }(\eta _2^0)$ is written as a sum
of double products of functions depending on disjoint values. Furthermore,
by separating the first summand from the right-hand side of the previous formula,
it follows:
\begin{align*}
\det \Lambda^{\prime \prime }(\eta _2^0)  &
=(-1)^{l(m+1)}R_{13}^{\alpha _1}(\eta _2^0)\cdots
R_{13}^{\alpha _l}(\eta _2^0)\det (\left\langle
\alpha _u|\alpha _v\right\rangle )_{u,v=1}^l\mu_{1,13}\cdots
\mu_{l,13}\\
& +\sum _{I\in S_l,I\neq \{ 1,\dotsc,l\} }
(-1)^{\frac{(m-l+1)+m}{2}l+|I|}(\det \Delta _I)(\det \Delta _I^\prime ),
\end{align*}
with which, one concludes the proof.
\end{proof}
\subsection{Generators for gauge-invariant Lagrangians}
Given a vector $\alpha \in \oplus ^N(\mathrm{ad}P)_x$ and an element
$u\in \pi ^{-1}(x)$, there exist unique elements
$A_1,\dotsc,A_N\in \mathfrak{g}$ such that
$\alpha =\left( (u,A_1)_{\mathrm{ad}},\dotsc,(u,A_N)_{\mathrm{ad}}\right) $.
Actually, $\alpha $ can be written as
$\alpha =\left( (u_1^\prime ,A_1^\prime )_{\mathrm{ad}},\dotsc,
(u_N^\prime ,A_N^\prime )_{\mathrm{ad}}\right) $, with $u_h^\prime \in \pi ^{-1}(x)$,
$A_h^\prime \in \mathfrak{g}$, $1\leq h\leq N$. As $G$ operates freely
and transitively on the fibre $\pi^{-1}(x)$, there exist unique elements
$g_1,\dotsc,g_N\in G$ such that $u_h^\prime =u\cdot g_h$, $1\leq h\leq N$.
Hence
\begin{align*}
\alpha & =\left( (u\cdot g_1,A_1^\prime )_{\mathrm{ad}},\dotsc,
(u\cdot g_N,A_{N}^\prime )_{\mathrm{ad}}\right) \\
& =\left( (u,\mathrm{Ad}_{g_1}A_1^\prime )_{\mathrm{ad}},\dotsc,
(u,\mathrm{Ad}_{g_{N}}A_{N}^\prime )_{\mathrm{ad}}\right) .
\end{align*}

If $I\colon \oplus ^N\mathfrak{g}\to \mathbb{R}$ is a polynomial
function invariant under the diagonal action induced by the adjoint
representation of $G$ on its Lie algebra $\mathfrak{g}$, then a function
$\tilde{I}\in C^\infty (\oplus ^N\mathrm{ad}P)$ can be associated by setting
\[
\begin{array}
[c]{l}
\tilde{I}\left( (u,A_1)_{\mathrm{ad}},\dotsc,(u,A_N)_{\mathrm{ad}}\right)
=I(A_1,\dotsc,A_N),\\
\forall A_1,\dotsc,A_N\in \mathfrak{g},
\end{array}
\]
as the formula above makes sense because it does not depend on the
representative chosen.
In fact, any other representative of the element
\[\alpha =\left(
(u,A_1)_{\mathrm{ad}},\dotsc,(u,A_N)_{\mathrm{ad}}\right)
\]
is of the form
$\alpha =\left( (u\cdot g^{-1},\mathrm{Ad}_{g}A_1)_{\mathrm{ad}},\dotsc,
(u\cdot g^{-1},\mathrm{Ad}_{g}A_N)_{\mathrm{ad}}\right) $,
$\forall g\in G$.
As $I$ is invariant under the diagonal action we have
\[
\!
\begin{array}
[c]{rl}
\tilde{I}\left( (u\cdot g^{-1},\mathrm{Ad}_gA_1)_{\mathrm{ad}}
,\dotsc,(u\cdot g^{-1},\mathrm{Ad}_{g}A_N)_{\mathrm{ad}}\right) = &
I(\mathrm{Ad}_gA_1,\dotsc,\mathrm{Ad}_gA_N)\\
= & I(A_1,\dotsc,A_{N})\\
= & \tilde{I}\left( (u,A_1)_{\mathrm{ad}}
,\dotsc,(u,A_{N})_{\mathrm{ad}}\right)  .
\end{array}
\]
\begin{theorem}
\label{th2}
Assume the group $G$ is connected and semisimple.
If $(U;x^i)_{i=1}^{n}$ is a coordinate system such that $P$
is trivial over $U$, then a $\mathrm{Gau}P|_U$-equivariant
vector-bundle isomorphism
$\Psi \colon \bigwedge^2T^\ast U\otimes \mathrm{ad}P|_U
\to \oplus ^N\mathrm{ad}P|_U$, $N=\frac{1}{2}n(n-1)$,
is defined by
\begin{equation}
\begin{array}
[c]{rl}
\Psi (\eta _2)= & \left( \eta _2\left( \left(
\tfrac{\partial }{\partial x^1}\right) _x,
\left( \tfrac{\partial }{\partial x^2}\right) _x\right) ,\dotsc,
\eta _2\left( \left( \tfrac{\partial }{\partial x^{n-1}}\right) _x,
\left( \tfrac{\partial }{\partial x^n}\right) _x\right) \right) ,
\medskip \\
& \eta _2\in \bigwedge^2T_x^\ast U\otimes(\mathrm{ad}P)_x,\;x\in U.
\end{array}
\label{Psi}
\end{equation}
There exists a finite system of generators $I^i$, $1\leq i\leq\nu$,
of the algebra of polynomial functions $\mathcal{P}(\oplus ^N\mathfrak{g})^G$
invariant under the diagonal action induced by the adjoint representation
of $G$ on $\mathfrak{g}$ and the functions $\tilde{I}^i\circ \Psi $,
$1\leq i\leq \nu $, generate the algebra
$C^\infty (\bigwedge^2T^\ast U\otimes \mathrm{ad}P|_{U})^{\mathrm{Gau}P|_U}$
over $C^\infty (U)$ differentiably.

Finally, if $(U_\alpha;x_\alpha^i)$, $1\leq i\leq n$, $\alpha \in A$, is
a coordinate system such that $\cup _{\alpha \in A}U_\alpha =M$,
$P$ is trivial over every $U_\alpha$, and $(\eta _\alpha )$, $\alpha \in A$,
is a partition of unity subordinate to $(U_\alpha )$, $\alpha \in A$,
then the functions
$J^i=\sum _{\alpha \in A}\eta _\alpha (\tilde{I}_\alpha ^i\circ \Psi )$,
$1\leq i\leq\nu$, generate the algebra
$C^\infty (\bigwedge^2T^\ast U\otimes \mathrm{ad}P)^{\mathrm{Gau}P}$
over $C^\infty (M)$ differentiably.
\end{theorem}
\begin{proof}
From the very definition of $\Psi $, it follows that the map $\Psi _x$
induced on every fibre $\bigwedge^2T_x^\ast U\otimes(\mathrm{ad}P)_x$,
$x\in U$, is $\mathbb{R}$-linear.

If $\eta _2=\sum _{1\leq i<j\leq n}(dx^i)_x\wedge (dx^j)_x\otimes A_{ij}$,
$A_{ij}\in (\mathrm{ad}P)_x$, belongs to $\ker \Psi _x$, then
$A_{ij}=0$ for $1\leq i<j\leq n$. Hence $\Psi _x$ is injective and since
the vector bundles $\bigwedge^2T^\ast M\otimes \mathrm{ad}P$ and
$\oplus ^N\mathrm{ad}P$ have te same rank, we conclude that $\Psi $
is an isomorphism. Moreover, every $\phi \in \mathrm{Gau}P$ acts
on $\mathrm{ad}P$ as explained in \S \ref{aut and Aut}, namely
$\Phi _{\mathfrak{g}}((u,v)\operatorname{mod}G)
=(\Phi (u),v)\operatorname{mod}G$,
$\forall u\in P$, $\forall v\in \mathfrak{g}$, and acts trivially
on $\bigwedge^2T^\ast M$, thus proving that $\Psi $ is
$\mathrm{Gau}P|_U$-equivariant.

As $G$ is assumed to be connected and semisimple, according to Weyl's Theorem
every (finite dimensional) linear representation of $G$ is completely
reducible. Hence by virtue of Hilbert-Nagata Theorem there exists a finite
system of generators $I^i$, $1\leq i\leq\nu$, of the algebra of polynomial
functions $\mathcal{P}(\oplus ^N\mathfrak{g})^G$ invariant under the
diagonal action induced by the adjoint representation of $G$ on $\mathfrak{g}$,
as said in the statement. If $I\colon \oplus ^N\mathfrak{g}
\to \mathbb{R}^\nu $ is the map whose components are $I^1,\dotsc,I^\nu$,
then by virtue of the main result in \cite{Luna} it follows that for every
$f\in C^\infty (\oplus ^N\mathfrak{g})^G$ there exists
$g\in C^\infty (\mathbb{R}^{\nu})$ such that $f=g\circ I$. Moreover,
as $P|_U$ is trivial by virtue of the hypothesis we can choose
a trivialization $\tau\colon P|_U\to U\times G$, i.e., $\tau $
is an isomorphism of principal $G$-bundles such that
$\mathrm{pr}_1\circ \tau =\pi$, which induces an isomorphism
of vector bundles over $U$, $\tau _{\mathrm{ad}}\colon \oplus ^N\mathrm{ad}P|_U
\to U\times \oplus ^N\mathfrak{g}$. Hence every
$F\in C^\infty (\bigwedge^2T^\ast U\otimes \mathrm{ad}P|_U)^{\mathrm{Gau}P}$
can be written as $F=(g\circ I)\circ \mathrm{pr}_2\circ \tau _{\mathrm{ad}}
\circ \Psi $, where $\mathrm{pr}_2\colon U\times\oplus ^N\mathfrak{g}
\to \oplus ^N\mathfrak{g}$ denotes the canonical projection onto the second factor.

Finally, as $\operatorname*{support}\eta _\alpha\subset U_\alpha $ and
$\eta _\alpha \in C^\infty (M)$, it follows that
$\eta _\alpha (\tilde{I}_\alpha ^i\circ \Psi )$ is globally defined
for $1\leq i\leq\nu$, and taking account of the fact that the vector fields
$\chi _\alpha $, $1\leq \alpha \leq m$, in \eqref{inv_ad} spanning
the distribution $\mathcal{D}$, are vertical with respect to the natural
projection $\bigwedge^2T^\ast M\otimes \mathrm{ad}P\to M$, we can conclude
that $X(J^i)=0$ for every $1\leq i\leq\nu $ and every
$X\in \Gamma (M,\mathcal{D})$, thus finishing the proof.
\end{proof}
\section{Remarks and Examples}
\begin{remark}
\label{remark1}
Let $G$ be a connected Lie group. As $M$ is also assumed to be
connected and oriented, if $n=2$, then $\bigwedge\nolimits^2T^\ast M$
is a trivial bundle of rank $1$; hence
$\bigwedge\nolimits^2T^\ast M\otimes \mathrm{ad}P$ is isomorphic
to the adjoint bundle by means of the map
$\bigwedge\nolimits^2T^\ast M\otimes \mathrm{ad}P\to \mathrm{ad}P$,
$\mathbf{v}_x\otimes X\mapsto X$, $\forall X\in (\mathrm{ad}P)_x$,
$\mathbf{v}$ being the volume form on $M$. As $G$ is connected, the first
integrals of $\mathcal{D}$ coincide with the $C^\infty $ functions on
$\mathrm{ad}P$ invariant under the adjoint representation of $G$ on
$\mathfrak{g}$, namely \ $C^\infty (\mathrm{ad}P)^G
=\left\{ f\in C^\infty (\mathrm{ad}P):X(f)=0,
\forall X\in \mathcal{D}\right\} $.
If $p\in S^\bullet (\mathfrak{g}^\ast )$ is a polynomial invariant
under the adjoint representation of $G$ on its Lie algebra and if $G$
is a semisimple complex Lie group, then we can consider the function
$\tilde{p}\in C^\infty (\mathrm{ad}P)$ defined as above
and Chevalley's theorem (e.g., see \cite[Theorem 4.9.3]{Var})
ensures the existence of $l$ homogeneous algebraically independent polynomials
$p_1,\dotsc,p_l$ such that the algebra of polynomial functions on $\mathfrak{g}$
that are invariant with respect to the adjoint representation of $G$
on $\mathfrak{g}$, is isomorphic to $\mathbb{C}[p_1,\dotsc,p_l]$, i.e.,
$S^\bullet(\mathfrak{g}^\ast )^G=\mathbb{C}[p_1,\dotsc,p_l]$.
If $\varphi \colon \mathrm{ad}P\to M\times \mathbb{C}^l$ is the map given by
\[
\varphi ((u,A)_{\mathrm{ad}})
=\left( \pi (u),\tilde{p}_1((u,A)_{\mathrm{ad}}),\dotsc,
\tilde{p}_l((u,A)_{\mathrm{ad}})\right) ,
\]
then according to \cite[\textsc{Theorem}. 3)]{Luna} we have
$C^\infty (\mathrm{ad}P)^G=\varphi ^\ast C^\infty (M\times\mathbb{C}^l)$
and by applying Frobenius theorem we conclude that
$C^\infty (\bigwedge\nolimits^2T^\ast M\otimes\mathrm{ad}P)^G$
is generated by the $l$ functions $\tilde{p}_1,\dotsc,\tilde{p}_l$.

This fact explains why equations of Yang-Mills type on a surface admit a
geometric treatment; see \cite{AB}.
\end{remark}
\begin{remark}
\label{remark2}
The map $\pi ^{(N)}\colon P^{(N)}=P\times _M\overset{(N}{\ldots }
\times_mP\to M$ is a principal fibre bundle with structure group
$G^N$, the adjoint bundle of which can be identified
to $\mathrm{ad}P^{(N)}=\oplus ^N\mathrm{ad}P$. If $G$ is complex
and semisimple, then $G^N$ also is, and we have $\dim G^N=Nm$,
$\operatorname*{rank}\mathfrak{g}^{N}=Nl$,
$l=\operatorname*{rank}\mathfrak{g}$. Hence Chevalley's
theorem can be applied to the adjoint representation of $G^N$ on its Lie
algebra, thus deducing that its ring of invariants admits a basis of $Nl$
algebraically independent homogeneous polynomials that can be constructed
from a basis $p_1,\dotsc,p_l$ of Chevalley for $\mathfrak{g}$ by defining
the polynomial $p_i^h\colon \mathfrak{g}^N\to \mathbb{C}$ by
$p_i^h(A_1,\dotsc,A_N)=p_i(A_h)$, $1\leq i\leq l$, $1\leq h\leq N$.
In particular, such polynomials are also invariant under the diagonal
action induced by the adjoint representation of $G$ on $\mathfrak{g}^N$;
but the polynomials $(p_i^h)_{1\leq i\leq l}^{1\leq h\leq N}$
do not generate in general the ring of invariants for the diagonal action.
In fact, if $n\geq 3$, then we know that the generic rank of the distribution
$\mathcal{D}$ is $m$; hence, the maximum number of functionally independient
gauge-invariant functions is $mN-m=m(N-1)$. As the number of polynomials
$p_i^h$ is $Nl$, if such polynomials span the ring of gauge invariants
it should be $Nl\geq m(N-1)$, i.e., $\frac{l}{m}\geq 1-\frac{1}{N}$,
and this inequality never occurs in the semisimple case,
as in this case we have $m>2l$. Actually, as dimension
and rank are aditive, it suffice to prove the last inequality
for simple algebras:
\[
\begin{array}
[c]{rlrl}
\dim\mathfrak{sl}(l+1,\mathbb{C})= & l(l+2)>2l, & \dim\mathfrak{so}
(2l+1,\mathbb{C})= & l(2l+1)>2l,\\
\dim\mathfrak{sp}(2l,\mathbb{C})= & l(2l+1)>2l, & \dim\mathfrak{so}
(2l,\mathbb{C})= & l(2l-1)>2l,
\end{array}
\]
as in the fourth case me must have $l\geq2$; and for the exceptional algebras:
\[
\begin{array}
[c]{rlrlrl}
\dim\mathfrak{e}_6= & 78>12, & \dim\mathfrak{e}_7= & 133>14, &
\dim\mathfrak{e}_8= & 248>16,\\
\dim\mathfrak{f}_4= & 52>8, & \dim\mathfrak{g}_2= & 14>4. &  &
\end{array}
\]
\end{remark}
\begin{example}
\label{example0}
If $\dim M=n=3$ y $\mathfrak{g}=\mathfrak{sl}(2,\mathbb{R})$,
then $m=3$, $l=1$. The basic invariant is
$p(X)=x_{12}x_{21}+(x_{11})^2=-\det (X)$, where
\[
X=\left(
\begin{array}
[c]{cc}
x_{11} & x_{12}\\
x_{21} & -x_{11}
\end{array}
\right) \in \mathfrak{sl}(2,\mathbb{R}).
\]
The maximum number of functionally independent $G$-invariant functions over
the curvature bundle $\bigwedge\nolimits^2T^{\ast}(M)\otimes\mathrm{ad}
P\underset{\text{locally}}{\cong}\oplus^3\mathrm{ad}P$ is $3\tbinom{3}
{2}-3=6$ in this case.

To the quadratic polynomial $p$ corresponds a unique symmetric bilinear form
$s_{p}\colon\mathfrak{sl}(2,\mathbb{R})\times\mathfrak{sl}(2,\mathbb{R}
)\to \mathbb{R}$ obtained by polarization
\[
p(X,Y)=\tfrac{1}{2}\left[ p(X+Y)-p(X)-p(Y)\right]
=x_{11}y_{11}+\tfrac{1}{2}x_{12}y_{21}+\tfrac{1}{2}x_{21}y_{12},
\]
from which the following invariants on $\oplus^3\mathfrak{sl}(2,\mathbb{R})$
follow:
\[
\begin{array}
[c]{rlrl}
I_{11}(X,Y,Z)= & s_{p}(X,X), & I_{22}(X,Y,Z)= & s_{p}(Y,Y),\\
I_{33}(X,Y,Z)= & s_{p}(Z,Z), & I_{12}(X,Y,Z)= & s_{p}(X,Y),\\
I_{13}(X,Y,Z)= & s_{p}(X,Z), & I_{23}(X,Y,Z)= & s_{p}(Y,Z).
\end{array}
\]
As a calculation shows, these $6$ functions are functionally independent in
the dense open subset defined by
\begin{multline*}
0\neq z_{11}(x_{11}y_{21}-x_{21}y_{11})\cdot \\
\left( x_{11}y_{12}z_{21}-x_{11}y_{21}z_{12}-x_{12}y_{11}z_{21}
+x_{12}y_{21}z_{11}+x_{21}y_{11}z_{12}-x_{21}y_{12}z_{11}\right) .
\end{multline*}
By using the local method in Theorem \ref{th2} by means of the map
\eqref{Psi}, it follows:
\[
\begin{array}
[c]{rl}
I_{11}\circ \Psi= & R_{12}^1R_{12}^3+(R_{12}^2)^2,
\smallskip\\
I_{12}\circ \Psi= & R_{12}^2R_{13}^2+\tfrac{1}{2}
\left( R_{12}^1R_{12}^3+R_{12}^1R_{12}^3\right) ,
\smallskip\\
I_{13}\circ \Psi= & R_{12}^2R_{23}^2+\tfrac{1}{2}
\left( R_{23}^1R_{12}^3+R_{12}^1R_{23}^3\right) ,
\smallskip \\
I_{22}\circ \Psi= & R_{13}^3R_{13}^1+(R_{13}^2)^2,
\smallskip\\
I_{23}\circ \Psi= & R_{13}^2R_{23}^2
+\tfrac{1}{2}\left(  R_{23}^1R_{13}^3+R_{13}^1R_{23}^3\right)  ,
\smallskip\\
I_{33}\circ\Psi= & R_{23}^1R_{23}^3+(R_{23}^2)^2.
\end{array}
\]
\end{example}
\begin{example}
\label{example1}
If $(e_{\alpha \beta })_{\alpha, \beta =1}^m$ denotes the
standard basis for $\mathfrak{gl}(m,\mathbb{R})$, then we have
$[e_{\alpha \beta },e_{\rho \sigma }]
=\delta _{\beta \rho }e_{\alpha \sigma }
-\delta_{\alpha \sigma }e_{\rho\beta}$. If
$\mathfrak{g}=\left\langle e_{12},e_{13},e_{23}\right\rangle $
denotes the Heisenberg Lie algebra and $B_1=e_{12}$,
$B_2=e_{13}$, $B_3=e_{23}$, then $[B_1,B_2]=0$, $[B_1,B_3]=-B_2$,
$[B_2,B_3]=0$, or equivalently $c_{12}^\alpha =c_{23}^\alpha =0$,
$c_{13}^\alpha =-\delta _2^\alpha $, $1\leq \alpha \leq3$;
it thus follows: $\chi _1={\textstyle\sum\nolimits_{i<j}}
R_{ij}^3\tfrac{\partial }{\partial R_{ij}^2}$, $\chi _2=0$,
$\chi _3=-{\textstyle\sum\nolimits_{i<j}}
R_{ij}^1\tfrac{\partial }{\partial R_{ij}^2}$.
We distinguish two cases:
1st) If $n=2$, then
$\chi _1=R_{12}^3\tfrac{\partial }{\partial R_{12}^2}$,
$\chi _3=-R_{12}^1\tfrac{\partial }{\partial R_{12}^2}$,
and these two vectors are proportional on a dense open subset.
Hence the generic rank of $\mathcal{D}$ is $1$ in this case;
2nd) If $n\geq 3$, then the matrix of $\chi _1$ and $\chi _2$
in the basis $\tfrac{\partial }{\partial R_{ij}^2}$,
$1\leq i<j\leq n$, is as follows:
\[
\left(
\begin{array}
[c]{cccccccc}
R_{12}^3 & \ldots & R_{1n}^3 & R_{23}^3 & \ldots & R_{2n}^3 & \ldots &
R_{n-1,n}^3\\
-R_{12}^1 & \ldots & -R_{1n}^1 & -R_{23}^1 & \ldots & -R_{2n}^1 &
\ldots & -R_{n-1,n}^1
\end{array}
\right) ,
\]
and on the dense open subset on which at least one of the determinants
\[
{\small \left\vert \!
\begin{array}
[c]{cc}
\!R_{hi}^3 & \!R_{jk}^3\\
\!-R_{hi}^1 & \!-R_{jk}^1
\end{array}
\!\right\vert }, \quad
1\leq h<i\leq j<k\leq n,
\]
does not vanish, the rank of $\mathcal{D}$ is $2$ in this case.
Note that the algebra $\mathfrak{g}$ under consideration is not semisiple.
\end{example}
\begin{example}
\label{example2}
Let us consider $\mathfrak{g}=\mathfrak{sl}(2,\mathbb{R})$
with its standard basis $B_1=e_{21}$, $B_2=e_{11}-e_{22}$, $B_3=e_{12}$;
$[B_1,B_2]=2B_1$, $[B_1,B_3]=-B_2$, $[B_2,B_3]=2B_3$, i.e.,
$c_{12}^\alpha =2\delta _1^\alpha $, $c_{13}^\alpha =-\delta _2^\alpha $,
$c_{23}^\alpha =2\delta _3^\alpha $, $1\leq \alpha \leq 3$. If $n=2$,
then we have
\[
\begin{array}
[c]{rlrl}
\chi _1= & -2R_{12}^2\tfrac{\partial }{\partial R_{12}^1}
+R_{12}^3\tfrac{\partial }{\partial R_{12}^2}, & \chi _2
= & 2R_{12}^1\tfrac{\partial }{\partial R_{12}^1}
-2R_{12}^3\tfrac{\partial }{\partial R_{12}^3},\\
\chi_3= & -R_{12}^1\tfrac{\partial }{\partial R_{12}^2}
+2R_{12}^2
\tfrac{\partial }{\partial R_{12}^3}, & 0= & R_{12}^1\chi_1
+R_{12}^2\chi _2+R_{12}^3\chi_3.
\end{array}
\]
Consequently, the generic rank of $\mathcal{D}$ is $2$ in this case.
If $n\geq 3$, then the components $\chi _1^\prime $, $\chi _2^\prime
$, $\chi _3^\prime $ of the vector fields $\chi _1$, $\chi _2$,
$\chi _3$, respectively, in the subspace generated by
$\frac{\partial }{\partial R_{ij}^\alpha }$, $1\leq \alpha \leq 3$,
$1\leq i<j\leq 3$, are
\begin{eqnarray*}
\chi _1^\prime &=&\!\!\!\!-2R_{12}^2
\tfrac{\partial }{\partial R_{12}^1}
-2R_{13}^2\tfrac{\partial }{\partial R_{13}^1}-2R_{23}^2
\tfrac{\partial }{\partial R_{23}^1}+R_{12}^3
\tfrac{\partial }{\partial R_{12}^2}+R_{13}^3
\tfrac{\partial }{\partial R_{13}^2}+R_{23}^3
\tfrac{\partial }{\partial R_{23}^2}, \\
\chi _2^\prime &=&\!\!\!\!2R_{12}^1
\tfrac{\partial }{\partial R_{12}^1}
+2R_{13}^1\tfrac{\partial }{\partial R_{13}^1}+2R_{23}^1
\tfrac{\partial }{\partial R_{23}^1}-2R_{12}^3
\tfrac{\partial }{\partial R_{12}^3}-2R_{13}^3
\tfrac{\partial }{\partial R_{13}^3}-2R_{23}^3
\tfrac{\partial }{\partial R_{23}^3}, \\
\chi _3^\prime &=&\!\!\!\!-R_{12}^1
\tfrac{\partial }{\partial R_{12}^2}-R_{13}^1
\tfrac{\partial }{\partial R_{13}^2}-R_{23}^1
\tfrac{\partial }{\partial R_{23}^2}+2R_{12}^2
\tfrac{\partial }{\partial R_{12}^3}
+2R_{13}^2\tfrac{\partial }{\partial R_{13}^3}
+2R_{23}^2\tfrac{\partial }{\partial R_{23}^3},
\end{eqnarray*}
or in matrix notation:
\begin{equation*}
\begin{array}{lccccccccc}
& \!\!\!\!\!\tfrac{\partial }{\partial R_{12}^1}
& \!\!\tfrac{\partial }{\partial R_{13}^1}
& \!\!\tfrac{\partial }{\partial R_{23}^1}
& \!\!\tfrac{\partial }{\partial R_{12}^2}
& \!\!\tfrac{\partial }{\partial R_{13}^2}
& \!\!\tfrac{\partial }{\partial R_{23}^2}
& \!\!\tfrac{\partial }{\partial R_{12}^3}
& \!\!\tfrac{\partial }{\partial R_{13}^3}
& \!\!\tfrac{\partial }{\partial R_{23}^3} \\
&  &  &  &  &  &  &  &  &  \\
\chi _1^\prime : & \!\!\!\!\!-2R_{12}^2 & \!\!-2R_{13}^2 &
\!\!-2R_{23}^2 & \!\!R_{12}^3 & \!\!R_{13}^3 & \!\!R_{23}^3 & \!\!0
& \!\!0 & \!\!0 \\
\chi _2^\prime : & \!\!\!\!\!2R_{12}^1 & \!\!2R_{13}^1 &
\!\!2R_{23}^1 & \!\!0 & 0 & \!\!0 & \!\!-2R_{12}^3 & \!\!-2R_{13}^3 &
\!\!-2R_{23}^3 \\
\chi _3^\prime : & \!\!\!\!\!0 & \!\!0 & \!\!0 & \!\!-R_{12}^1 &
\!\!-R_{13}^1 & \!\!-R_{23}^1 & \!\!2R_{12}^2 & \!\!2R_{13}^2 &
\!\!2R_{23}^2
\end{array}
\end{equation*}
and the determinant
\[
\left\vert
\begin{array}
[c]{ccc}
-2R_{12}^2 & -2R_{13}^2 & R_{12}^3\\
2R_{12}^1 & 2R_{13}^1 & 0\\
0 & 0 & -R_{12}^1
\end{array}
\right\vert =4R_{12}^1\left( R_{13}^1R_{12}^2-R_{12}^1R_{13}^2\right)
\]
does not vanish identically. Accordingly, for $n\geq3$ the generic rank
of $\mathcal{D}$ is $3$.
\end{example}
\begin{example}
\label{example3}If $n=3$, $\mathfrak{g}=\mathfrak{so}(3,\mathbb{R})$,
then $m=3$, $l=1$, and by considering the standard basis
$B_1=e_{12}-e_{21}$, $B_2=e_{13}-e_{31}$, $B_3=e_{23}-e_{32}$
we have $X=xB_1+yB_2+zB_{23}$,
$\det(tI_3-X)=t^3+\left( x^2+y^2+z^2\right) t$.
Thus, the basic invariant is $p(X)=x^2+y^2+z^2$. Hence
\[
\begin{array}
[c]{rl}
I_{ii}(X_1,X_2,X_3)= & (x_i)^2+(y_i)^2+(z_i)^2,\\
i= & 1,2,3,\\
I_{12}(X_1,X_2,X_3)= & x_1x_2+y_1y_2+z_1z_2,\\
I_{13}(X_1,X_2,X_3)= & x_1x_3+y_1y_3+z_1z_3,\\
I_{23}(X_1,X_2,X_3)= & x_2x_3+y_2y_3+z_2z_3,
\end{array}
\]
and one can check that these invariants are
functionally independent and therefore build a
system of generators for the ring of (local) invariants.
\end{example}
\begin{example}
\label{example4}If $n=3$, $\mathfrak{g}=\mathfrak{so}(4,\mathbb{R})$,
then $m=6$, $l=2$ and the characteristic polynomial of
$X=\sum_{1\leq i<j\leq 4}x_{ij}(e_{ij}-e_{ji})$ is
\[
\begin{array}
[c]{rl}
\det(tI_{4}-X)= & t^{4}+p_1(X)t^2+p_2(X)^2,\\
p_1(X)= & \sum\nolimits_{1\leq i<j\leq4}(x_{ij})^2,\\
p_2(X)= & x_{12}x_{34}+x_{14}x_{23}-x_{13}x_{24}.
\end{array}
\]
Proceeding as above, by polarizing $p_1$ and $p_2$, we obtain $12$
generically independent invariant functions, which coincides with the corank
of the ditribution $\mathcal{D}$\ in this case.
\end{example}
\end{document}